\documentclass[11pt]{article}

\usepackage[margin=1in]{geometry}
\usepackage{amsthm,amsmath,amsfonts,amssymb}
\usepackage{amsrefs}\renewcommand{\eprint}[1]{\href{http://arxiv.org/abs/#1}{#1}}
\usepackage{url}
\usepackage[bookmarksnumbered,colorlinks,linkcolor=blue,citecolor=blue,urlcolor=blue,
				pdftitle={Limitations on the simulation of non-sparse Hamiltonians},
				pdfauthor={Andrew M. Childs and Robin Kothari},
				pdfkeywords={}]{hyperref}
\usepackage{colonequals}
\usepackage{thm-restate}

\newtheorem{theorem}{Theorem}
\newtheorem{lemma}{Lemma}
\newtheorem{proposition}{Proposition}
\theoremstyle{definition}
\newtheorem{definition}{Definition}

\newcommand{\eq}[1]{(\ref{eq:#1})}
\renewcommand{\sec}[1]{Section~\ref{sec:#1}}
\newcommand{\thm}[1]{Theorem~\ref{thm:#1}}
\newcommand{\lem}[1]{Lemma~\ref{lem:#1}}
\newcommand{\prop}[1]{Proposition~\ref{prop:#1}}

\newcommand{\C}{{\mathbb{C}}}

\newcommand{\R}{{\mathbb{R}}}
\newcommand{\Z}{{\mathbb{Z}}}

\newcommand{\spn}{\mathop{\mathrm{span}}}
\newcommand{\poly}{\mathop{\mathrm{poly}}}
\newcommand{\abs}{\mathop{\mathrm{abs}}}
\newcommand{\mcn}{\mathop{\mathrm{mcn}}}
\newcommand{\ii}{{\mathrm{i}}}
\newcommand{\inuar}{\in_{\mathrm{R}}}
\renewcommand{\o}{o}
\renewcommand{\O}{O}
\renewcommand{\(}{\left(}
\renewcommand{\)}{\right)}
\newcommand{\defeq}{\colonequals}

\newcommand{\norm}[1]{\|{#1}\|}

\renewcommand{\>}{\rangle}

\newcommand{\be}{\begin{equation}}
\newcommand{\ee}{\end{equation}}
\def\ba#1\ea{\begin{align}#1\end{align}}

\begin{document}


\title{Limitations on the simulation of non-sparse Hamiltonians}

\author{
\normalsize Andrew M.\ Childs\thanks{amchilds@uwaterloo.ca} \\[.5ex]
\small Department of Combinatorics \& Optimization \\
\small and Institute for Quantum Computing \\
\small University of Waterloo
\and
\normalsize Robin Kothari\thanks{rkothari@cs.uwaterloo.ca} \\[.5ex]
\small David R.\ Cheriton School of Computer Science \\
\small and Institute for Quantum Computing \\
\small University of Waterloo
}

\date{}
\maketitle

\begin{abstract}
The problem of simulating sparse Hamiltonians on quantum computers is well studied. The evolution of a sparse $N \times N$ Hamiltonian $H$ for time $t$ can be simulated using $\O(\norm{Ht} \poly(\log N))$ operations, which is essentially optimal due to a no--fast-forwarding theorem. Here, we consider non-sparse Hamiltonians and show significant limitations on their simulation. We generalize the no--fast-forwarding theorem to dense Hamiltonians, ruling out generic simulations taking time $\o(\norm{Ht})$, even though $\norm{H}$ is not a unique measure of the size of a dense Hamiltonian $H$. We also present a stronger limitation ruling out the possibility of generic simulations taking time $\poly(\norm{Ht},\log N)$, showing that known simulations based on discrete-time quantum walk cannot be dramatically improved in general. On the positive side, we show that some non-sparse Hamiltonians can be simulated efficiently, such as those with graphs of small arboricity.
\end{abstract}

\section{Introduction}
\label{sec:intro}

One of the primary applications of quantum computers is the simulation of quantum systems. Indeed, it was the apparent exponential time complexity of simulating quantum systems on a classical computer that led Feynman to propose the idea of quantum computation~\cite{Fey82}.

In addition to predicting the behavior of physical systems, Hamiltonian simulation has algorithmic applications.  For example, the implementation of a continuous-time quantum walk algorithm is a Hamiltonian simulation problem.  Examples of algorithms  that can be implemented using Hamiltonian simulation methods include unstructured search~\cite{FG96}, adiabatic optimization~\cite{FGGS00}, a quantum walk with exponential speedup over classical computation~\cite{CCDFGS03}, and the recent NAND tree evaluation algorithm~\cite{FGG07}.

In the Hamiltonian simulation problem, our goal is to implement the unitary operator $e^{-\ii Ht}$ for some given Hamiltonian $H$ and time $t$.  We say that a Hamiltonian $H$ acting on an $N$-dimensional quantum system can be simulated efficiently if there is a quantum circuit using $\poly(\log N,t,1/\epsilon)$ one- and two-qubit gates that approximates (with error at most $\epsilon$) the evolution according to $H$ for time $t$. (Of course, we can rescale $t$ by rescaling $H$, so the complexity of simulating $H$ for time $t$ must also depend on some measure of the size of $H$, as discussed in more detail below.)

Efficient simulations are known for various classes of Hamiltonians.  For example, a Hamiltonian for a system of qubits can be simulated efficiently whenever it is \emph{local}, meaning that it is a sum of terms, each of which acts on a constant number of qubits~\cite{Llo96}.  More generally, a Hamiltonian $H$ can be simulated efficiently if it is \emph{sparse} (i.e., has only $\poly(\log N)$ nonzero entries per row) and \emph{efficiently row-computable} (i.e., there is an efficient means of computing the indices and the matrix elements of the nonzero entries in any given row) \cite{AT03}. 

These conditions lead to a convenient black-box formulation of the problem, in which a black box can be queried with a row index $j$ and an index $i$ to obtain the $i^\mathrm{th}$ nonzero entry in the $j^\mathrm{th}$ row. This black box can be implemented efficiently provided that $H$ is efficiently row-computable. A series of results has decreased the number of black-box queries, in terms of $N$, from the original $\O(\log^9 N)$~\cite{AT03}, to $\O(\log^2 N)$~\cite{Chi04}, to $\O(\log^* N)$~\cite{BACS05}.  In particular, Berry, Ahokas, Cleve, and Sanders~\cite{BACS05} present an almost linear-time algorithm for simulating sparse Hamiltonians with query complexity
\be
(\log^* N) d^4 \norm{Ht} {\left( \frac{\norm{d^2 Ht}} {\epsilon} \right)}^{\o(1)} ,
\label{eq:sparsecomplexity}
\ee
where $d$ is the maximum number of nonzero entries in any row and $\epsilon$ is the maximum error permitted in the final state (quantified in terms of trace distance).

The dependence of \eq{sparsecomplexity} on the simulation time is nearly optimal, since it is not possible to simulate a general sparse Hamiltonian for time $t$ using $o(t)$ queries.  Intuitively, there is no generic way to fast-forward through the time evolution of quantum systems.  More formally,

\begin{theorem}[{No--fast-forwarding theorem~\cite[Theorem 3]{BACS05}}]
\label{thm:fastfwd}
For any positive integer $N$ there exists a row-computable sparse Hamiltonian $H$ with $\norm{H}=1$ such that simulating the evolution of $H$ for time $t = \pi N/2$ within precision $1/4$ requires at least $N/4$ queries to $H$.
\end{theorem}

More recently, methods have been presented for simulating a Hamiltonian $H$ that is not necessarily sparse.  Of course, we do not expect to efficiently simulate a general Hamiltonian, simply because there are too many Hamiltonians to consider (just as we cannot hope to efficiently implement a general unitary operation \cite{Kni95}).  However, we can conceivably efficiently simulate non-sparse Hamiltonians with a suitable concise description.  In particular, by applying phase estimation to a discrete-time quantum walk derived from $H$, one can simulate $H$ for time $t$ in a number of walk steps that grows only linearly with $t$~\cite{Chi08}.  More precisely, we have

\begin{theorem}\label{thm:dense}
For any Hermitian matrix $H$, there is a discrete-time quantum walk on the graph of nonzero entries of $H$ such that $e^{-\ii Ht}$ can be simulated with error at most $\delta$ using $\O(\norm{\abs(Ht)}/\sqrt\delta)$ steps of the walk, where $\abs(H)$ is the matrix with entries $\abs(H)_{jk} = |H_{jk}|$.
\end{theorem}

Of course, to apply this result, we must implement the discrete-time quantum walk derived from $H$.  This can be done efficiently for various concisely specified non-sparse Hamiltonians~\cite{Chi08}.  Note that the same theorem holds with $\norm{\abs(H)}$ replaced by $\norm{H}_1$ (a matrix norm defined in \sec{measures}); this quantity is generally larger than $\norm{\abs(H)}$, but the resulting walk may be easier to implement.

Notice that the overhead of this simulation is proportional not to the spectral norm $\norm{H}$, but to a measure of the size of $H$ that can be much larger when some entries of $H$ are negative (or more generally, complex).  This naturally raises the question of whether an improved simulation is possible.  In the present article, we examine this possibility.  Unfortunately, our main result is negative: there is no general Hamiltonian simulation algorithm that uses only $\poly(\norm{Ht}, \log N)$ steps (\thm{densehard}).

The remainder of this article is organized as follows.  In \sec{measures}, we introduce various matrix norms that arise when quantifying the complexity of Hamiltonian simulation and relate them to one another. We then move on to lower bounds for non-sparse Hamiltonians in \sec{norm}, where we describe how the no--fast-forwarding theorem can be modified to give a lower bound that depends on the spectral norm rather than various smaller measures of the size of a Hamiltonian.  Then, in \sec{stronger}, we present the main result, an example of a family of Hamiltonians with $\norm{\abs(H)} \gg \norm{H}$ that cannot be simulated in time $\poly(\norm{Ht}, \log N)$. We then turn to upper bounds in \sec{structured}, and investigate how certain structured Hamiltonians can be simulated in time $O(\norm{Ht})$---in particular, we give a positive result on the simulation of Hamiltonians whose graphs have small arboricity.\footnote{A graph is said to have \emph{arboricity} $k$ if its adjacency matrix can be written as the sum the adjacency matrices of $k$ forests, but not $k-1$ forests.}  Finally, we conclude in \sec{open} with a discussion of open problems.

\section{Measures of simulation complexity}
\label{sec:measures}

Upper and lower bounds on the complexity of simulating a Hamiltonian $H$ depend on some measure of the size of $H$.  Since $e^{-\ii Ht}$ depends only on the product $Ht$, the complexity of simulating $H$ for time $t$ is some function of $Ht$.  For example, the no--fast-forwarding theorem clearly cannot be circumvented by simply multiplying $H$ by a constant.  Similarly, simulation results such as those for sparse Hamiltonians, using $\norm{Ht}^{1+o(1)}$ operations, and \thm{dense}, using $\O(\norm{\abs(Ht)})$ operations, depend on various measures of the size of $Ht$.

In this section, we take a step back and consider properties of various measures of the size of $H$ that may play a role in the complexity of simulating it.  Let $\nu(Ht)$ be a function that measures the complexity of simulating $H$ for time $t$.  We can infer various properties of $\nu(\cdot)$ as follows.  Since it is trivial to simulate the identity operation, $\nu(0)=0$.  On the other hand, if $H \ne 0$, then it requires some work to simulate, so $\nu(H) > 0$.  It is also plausible to suppose that $\nu(tH) = |t| \nu(H)$.  We clearly have $\nu(Ht)\leq|t|\nu(H)$ for $t \in \Z$, since $Ht$ can be simulated using $t$ exact simulations of $H$.  On the other hand, the no--fast-forwarding theorem suggests that this is the best possible way to simulate $Ht$ in general.  Finally, since the Lie product formula can be used to simulate $H+K$ using simulations of $H$ and $K$, we expect that $\nu(H+K) \lessapprox \nu(H)+\nu(K)$ (up to the fact that a bounded-error simulation requires a slightly superlinear number of operations).

These properties are reminiscent of the axioms for matrix norms, suggesting that it may be reasonable to quantify the complexity of simulating $H$ in terms of some matrix norm $\nu(H)$.  Indeed, results on the simulation of sparse Hamiltonians are typically stated in terms of the spectral norm $\norm{H}$, and \thm{dense} also involves matrix norms.  We now introduce various matrix norms relevant to Hamiltonian simulation.

\begin{definition}[Spectral norm]
The spectral norm of a matrix $H$ is defined as
\be
\norm{H}\defeq \max_{v\neq 0} \frac{\norm{Hv}}{\norm{v}}=\max_{\norm{v}=1} \norm{Hv},
\ee
where $\norm{v}$ is the standard Euclidean vector norm defined as $\norm{v}\defeq \sqrt{\sum_i |v_i|^2}$.
\end{definition}

The spectral norm, also known as the operator norm or induced Euclidean norm, is equal to the largest singular value of the matrix. For Hermitian matrices it is also equal to the magnitude of the largest eigenvalue. This norm arises in the  complexity of sparse Hamiltonian simulation algorithms, and in \thm{dense} as the spectral norm of $\abs(H)$, the matrix with entries $\abs(H)_{jk} = |H_{jk}|$.

\begin{definition}[Induced 1-norm]
The induced 1-norm of a matrix $H$ is defined as
\be
\norm{H}_1\defeq \max_{v\neq 0} \frac{\norm{Hv}_1}{\norm{v}_1}=\max_j \sum_i |H_{ij}|,
\ee
where $\norm{v}_1$ is the vector 1-norm defined as $\norm{v}_1\defeq {\sum_i |v_i|}$.
\end{definition}

The induced 1-norm is equal to the maximum absolute column sum of the matrix. As mentioned in \sec{intro}, \thm{dense} holds with $\norm{\abs(H)}$ replaced by $\norm{H}_1$. This does not, however, lead to a superior simulation method since $\norm{H}_1 \geq \norm{\abs(H)}$, as shown in \lem{norms} below.

\begin{definition}[Maximum column norm]
The maximum column norm of a matrix $H$ is defined as
\be
\mcn(H)\defeq \max_{j} \sqrt{\sum_i |H_{ij}|^2} = \max_{v\neq 0} \frac{\norm{Hv}}{\norm{v}_1} = \max_{j} \norm{He_j},
\ee
where $e_j$ is the $j^{\mathrm{th}}$ column of the identity matrix.
\end{definition}

The maximum column norm is the maximum Euclidean norm of the columns of $H$. This norm appears in the complexity of an algorithm for simulating Hamiltonians whose graphs are trees~\cite[Theorem~4]{Chi08} and in the related \prop{arb} in \sec{structured}.

\begin{definition}[Max norm]
The max norm of a matrix $H$ is defined as
\be
\max(H)\defeq \max_{i,j} |H_{ij}|. 
\ee
\end{definition}

The max norm is just the largest entry of $H$ in absolute value. It is a matrix norm, and is typically much smaller than the other norms mentioned.

The following lemma relates the various norms introduced above.

\begin{lemma}
\label{lem:norms}
For any Hermitian matrix $H \in \C^{N \times N}$, we have the following inequalities:
\be
\max(H) \leq \mcn(H) \leq \norm{H} \leq \norm{\abs(H)} \leq \norm{H}_1 \leq \sqrt{N} \mcn(H) \leq N \max(H) \label{eq:norms}.
\ee
Furthermore, each of these inequalities is the best possible.
\end{lemma}

\begin{proof}
The first inequality follows from the fact that the maximum element in any column cannot be greater than the Euclidean norm of that column. We have
\be
\max(H)=\max_j\Big(\max_i |H_{ij}| \Big)\leq\max_{j}\sqrt{\sum_i |H_{ij}|^2}=\mcn(H).
\ee

The next inequality follows from the observation that $\mcn(H)$ is defined by a maximum over the standard basis vectors $e_j$, whereas $\norm{H}$ is defined by a maximum over all vectors with norm $1$, which contains the set of all $e_j$. Thus  
\be
\mcn(H)=\max_j\norm{He_j}\leq\max_{\norm{v}=1} \norm{Hv}=\norm{H}.
\ee

Using the triangle inequality with $\norm{H}=\max_{\norm{v}=1}(\sum_i|\sum_j H_{ij}v_j|^2)^{\frac{1}{2}}$, we get 
\be
\norm{H}\leq \max_{\norm{v}=1}\bigg(\sum_i\bigg|\sum_j |H_{ij}||v_j|\bigg|^2\bigg)^{\frac{1}{2}}=
\max_{\substack{\norm{v}=1\\ v_j\geq 0}}\bigg(\sum_i\bigg|\sum_j \abs(H)_{ij}v_j\bigg|^2\bigg)^{\frac{1}{2}}.
\ee
Now by maximizing over all $v$ with $\norm{v}=1$ instead of only those with $v_j\ge 0$, we get
\be
\max_{\substack{\norm{v}=1\\ v_j\geq 0}}\bigg(\sum_i\bigg|\sum_j \abs(H)_{ij}v_j\bigg|^2\bigg)^{\frac{1}{2}}\leq
\max_{\norm{v}=1}\bigg(\sum_i\bigg|\sum_j \abs(H)_{ij}v_j\bigg|^2\bigg)^{\frac{1}{2}}=
\norm{\abs(H)}.
\ee
The last inequality is actually an equality due to the Perron--Frobenius theorem.

Since $\abs(H)$ is a symmetric matrix, there is an eigenvector $z$ with eigenvalue equal in magnitude to $\norm{\abs(H)}$. Clearly this eigenvector satisfies $\norm{\abs(H)z}_1 = \norm{\abs(H)}\norm{z}_1$. Using this and maximizing over all nonzero vectors, we have
\be
\norm{\abs(H)}=\frac{\norm{\abs(H)}\norm{z}_1}{\norm{z}_1}=\frac{\norm{\abs(H)z}_1}{\norm{z}_1}\leq \max_{v\neq 0} \frac{\norm{\abs(H)v}_1}{\norm{v}_1}=\norm{\abs(H)}_1.
\ee
The inequality now follows from the fact that $\norm{H}_1=\norm{\abs(H)}_1$, since 
\be
\norm{\abs(H)}_1=\max_j \sum_i |\abs(H)_{ij}|=\max_j \sum_i |H_{ij}|=\norm{H}_1.
\ee

For the next inequality, we use the fact that $\norm{v}_1\leq\sqrt{N}\norm{v}$ for all vectors $v$. This can be proved using the Cauchy--Schwarz inequality, $|\<u,v\>| \leq \norm{u} \norm{v}$, by taking $u_i=v_i/|v_i|$.
Let $j_{\max}$ be the index $j$ that maximizes $\sum_i |H_{ij}|$. Thus $\norm{H}_1=\sum_i |H_{ij_{\max}}|=\norm{He_{j_{\max}}}_1$. Using these two inequalities, it follows that 
\be\label{eq:ineq1}
\norm{H}_1=\norm{He_{j_{\max}}}_1\leq \sqrt{N}\norm{He_{j_{\max}}} \leq \sqrt{N} \max_{j} \norm{He_j} = \sqrt{N} \mcn(H).
\ee

The last inequality is proved using the fact that for any $j$, $H_{ij} \leq \max_i H_{ij}$; thus
\be\label{eq:ineq2}
\mcn(H)=\max_{j}\sqrt{\sum_i |H_{ij}|^2}\leq \max_{j}\sqrt{N\max_i |H_{ij}|^2} = \sqrt{N} \max_{ij} |H_{ij}| = \sqrt{N} \max(H). 
\ee

For each of these inequalities, there is a matrix that achieves equality. The first four inequalities are saturated when $H$ is the identity matrix since the relevant norms are all equal to $1$. The last two inequalities are satisfied with equality when $H$ is the all-ones matrix (i.e., for all $i,j$, $H_{ij}=1$), since then $\norm{H}_1 = N$, $\mcn(H)=\sqrt{N}$, and $\max(H)=1$.
\end{proof}

Since \thm{dense} involves $\norm{\abs(H)}$, we would like to relate $\norm{\abs(H)}$ and $\norm{H}$. \lem{norms} gives $\norm{\abs(H)} \le \sqrt{N}\norm{H}$, which is also the best possible inequality between the two norms. For example, when $N$ is a power of $2$, the matrix $H=R^{\otimes \log N}$ achieves equality, where $R \defeq \left(\begin{smallmatrix}1 & 1 \\ 1 & -1\end{smallmatrix}\right)/\sqrt2$ is the Hadamard matrix. It has $\norm{H}=1$, but $\norm{\abs(H)}=\sqrt{N}$. This shows that \thm{dense} might not be as powerful as we would like, since for some Hamiltonians, the simulation method of \thm{dense} may be infeasible even when $\norm{H}$ is small.

Although the above inequalities cannot be tightened in general, there can of course be stronger relationships among the various norms for special classes of Hamiltonians.  For example, if $H$ is sparse, observe that the norms mentioned above can differ at most by a factor of $\poly(\log N)$.  Specifically, if $H$ is $k$-sparse (i.e., it has at most $k$ nonzero entries per row), then
\be
\max(H) \leq \mcn(H) \leq \norm{H} \leq \norm{\abs(H)} \leq \norm{H}_1 \leq \sqrt{k} \mcn(H) \leq k \max(H). 
\label{eq:sparse}
\ee

The first four inequalities are from \lem{norms}. The inequality $\norm{H}_1 \leq \sqrt{k} \mcn(H)$ follows from \eq{ineq1} using the fact that $\norm{v}_1\leq\sqrt{k}\norm{v}$ when $v$ has at most $k$ non-zero entries (this can be proved using the Cauchy--Schwarz inequality as before). The last inequality follows from \eq{ineq2} and the inequality $\sum_i |H_{ij}|^2 \leq k \max_i |H_{ij}|^2$, which holds when $H$ is $k$-sparse.  Thus the choice of norm for sparse matrices is quite flexible, since all the above-mentioned norms are equivalent up to polynomial factors.

In \sec{structured}, we discuss some more examples in which \lem{norms} can be strengthened, with emphasis on the implications for simulations.

\section{A no--fast-forwarding theorem for dense Hamiltonians}
\label{sec:norm}

The no--fast-forwarding theorem (\thm{fastfwd} above) establishes a lower bound for the simulation of sparse Hamiltonians. Although we stated the theorem with $\norm{H}=1$, any of the norms in \lem{norms} could have been used, since the Hamiltonian used in the proof of the no--fast-forwarding theorem is $2$-sparse, and by \eq{sparse} the norms differ at most by a factor of $2$.  In particular, the theorem could be restated with $\max(H)\leq 1$ or $\norm{H}_1 \leq 2$.

Since the choice of norm is unclear, it is conceivable that there are Hamiltonian simulation algorithms that run in time $\O(\max(Ht))$ or $\O(\mcn(Ht))$. To distinguish between the norms, we require a dense Hamiltonian. The aim of this section is use the proof techniques of \thm{fastfwd} to establish a similar theorem for dense Hamiltonians. In particular, we show that there does not exist an algorithm for simulating dense Hamiltonians in time $\O(\max(Ht))$ or $\O(\mcn(Ht))$. However, this does not appear to rule out $\poly(\norm{Ht})$ simulations, which we rule out in the next section.  Although \thm{densehard} in the next section is stronger than \thm{densefastfwd} below, we briefly present this straightforward generalization of the no--fast-forwarding theorem to show the extent of that approach as applied to the non-sparse case.

As in \thm{fastfwd}, we consider a black-box formulation of the problem of simulating dense Hamiltonians. There is a black box that can be queried with a row index $j$, which outputs the entire $j^{\mathrm{th}}$ row. Since the Hamiltonian is dense, this output can be exponentially large. This is not a problem, however, since our goal is to find a lower bound on query complexity, not time complexity. Even though each query takes exponential space, it counts as only one query. The black box used here is more powerful than the one in \thm{fastfwd}, so the lower bound proved below also carries over to the black box used in \thm{fastfwd}.

In terms of this black-box model, we have the following:

\begin{theorem}\label{thm:densefastfwd}
For any positive integer $N$, there exists a non-sparse Hamiltonian $H$ such that simulating the evolution of $H$ for time $t = \pi N/2$ within precision $1/4$ requires at least $N/4$ queries to $H$. This Hamiltonian has $\norm{H}=1$, $\mcn(H)=\Theta(1/\sqrt{N})$, and $\max(H)=\Theta(1/N)$. 
\end{theorem}

\begin{proof}
The main idea, as in the proof of \thm{fastfwd}~\cite{BACS05}, is to construct a Hamiltonian whose simulation for time $t = \pi N/2$ determines the parity of $N$ bits. Since we know that computing the parity of $N$ bits requires at least $N/2$ queries~\cites{BBCMW01,FGGS98}, this Hamiltonian cannot be simulated with $\o(N)$ queries. Moreover, we want this Hamiltonian to be non-sparse.

We start with a simple Hamiltonian $H_1$ whose graph is just a line with $N+1$ vertices. Consider the Hamiltonian acting on vectors $|i\>$ with $i\in \{0,\ldots,N\}$. The nonzero matrix entries of $H_1$ are $\<i\left|H_1\right|i+1\>=\<i+1\left|H_1\right|i\>=\sqrt{(N-i)(i+1)}/N$
for $i \in \{0,1,\ldots,N-1\}$. This Hamiltonian has $\norm{H_1}=1$, and simulating $H_1$ for $t = \pi N/2$ starting with the state $|0\>$ gives the state $|N\>$ (i.e., $e^{-iH_1t}|0\>=|N\>$).

Now, as in Ref.~\cite{BACS05}, consider a Hamiltonian $H_2$ generated from an $N$-bit string $S_0 S_1 \ldots S_{N-1}$. $H$ acts on vertices $|i,j\>$, with $i\in \{0,\ldots,N\}$ and $j\in \{0,1\}$. The nonzero matrix entries of this Hamiltonian are  
\be
\<i,j\left|H_2\right|i+1,j\oplus S_i\>=\<i+1,j\oplus S_i\left|H_2\right|i,j\>=\sqrt{(N-i)(i+1)}/N
\ee
for all $i$ and $j$. By construction, $|0,0\>$ is connected to either $|i,0\>$ or $|i,1\>$ for any $i$; it is connected to $|i,j\>$ if and only if $j=S_0 \oplus S_1 \oplus \ldots \oplus S_{i-1}$. Thus $|0,0\>$ is connected to either $|N,0\>$ or $|N,1\>$, and determining which is the case determines the parity of $S$. The graph of this Hamiltonian consists of two disjoint lines, one of which contains $|0,0\>$ and either $|N,0\>$ or $|N,1\>$ depending on the parity of $S$. Just as for $H_1$, starting with the state $|0,0\>$ and simulating $H_2$ for time $t = \pi N/2$ will give either $|N,0\>$ or $|N,1\>$, which determines the parity of $S$. Note that since $H_2$ is a permutation of $H_1 \oplus H_1$, $\norm{H_2}=\norm{H_1}=1$.

Finally, we construct the dense Hamiltonian $H$ that has the properties stated in the theorem. As before, $H$ is generated from an $N$-bit string $S_0 S_1 \ldots S_{N-1}$. $H$ acts on vertices $|i,j,k\>$, with $i\in \{0,\ldots,N\}$, $j\in \{0,1\}$, and $k \in \{0,N-1\}$. The nonzero entries of $H$ are given by
\be
\<i,j,k\left|H\right|i+1,j\oplus S_i,k' \>=\<i+1,j\oplus S_i,k'\left|H\right|i,j,k\>=\sqrt{(N-i)(i+1)}/N^2
\ee
for all $i$, $j$, $k$, and $k'$. The graph of $H$ is similar to that of $H_2$, except that for each vertex in $H_2$, there are now $N$ copies of it in $H$. This Hamiltonian is dense because it has $\Theta(N^2)$ vertices and each vertex is connected to all $N$ copies of its neighboring vertices, which gives at least $N$ nonzero entries in each row. 

Now we simulate the Hamiltonian starting from the uniform superposition over the copies of the $|0,0\>$ state, i.e., from the state $\frac{1}{\sqrt{N}}\sum_k|0,0,k\>$. The subspace $\spn\{\sum_k|i,j,k\>\}$ of uniform superpositions over the third register is an invariant subspace of this Hamiltonian.  Since the initial state lies in this subspace, the quantum walk remains in this subspace.  In other words, the quantum walk on this dense graph starting from the chosen state reduces to the quantum walk on $H_2$ starting from the $|0,0\>$ state.

Now, just as before, the parity of $S$ can be determined by simulating $H$ for time $t = \pi N/2$. This gives the lower bound of $N/2$ queries.

To calculate the norms of this Hamiltonian, we observe that $H=H_2\otimes J/N$, where $J$ is the all-ones matrix of size $N \times N$. This gives $\norm{H}=\norm{H_2} \cdot \norm{J}/N = 1$. Direct computation shows that $\max(H)=\Theta(1/N)$ and $\mcn(H)=\Theta(1/\sqrt{N})$.
\end{proof}

This theorem rules out algorithms that make only $\O(\max(Ht))$ or $\O(\mcn(Ht))$ queries, since for this Hamiltonian  with $t = \pi N/2$ we have $\max(Ht)=\Theta(1)$ and $\mcn(Ht)=\Theta(\sqrt{N})$, both of which are disallowed by the lower bound of $\Omega(N)$.  However, this does not distinguish between $\norm{H}$ and $\norm{\abs(H)}$ (since $\abs(H)=H$), and in fact $\norm{H}_1 \sim 1$ as well.  In the next section we construct examples with $\norm{\abs{H}} \gg \norm{H}$ in order to show that a general simulation using $O(\norm{Ht})$ steps (or even $\poly(\norm{Ht}, \log N)$ steps)  is not possible.

\section{A stronger limitation for dense Hamiltonians}
\label{sec:stronger}

As discussed in \sec{intro}, there are dense Hamiltonian simulation algorithms that use $\O(\norm{\abs(Ht)})$ or $\O(\norm{Ht}_1)$ steps of a discrete-time quantum walk. However, in light of the no--fast-forwarding theorem for dense Hamiltonians, it might be reasonable to hope that dense Hamiltonians can be simulated in $\O(\norm{Ht})$ steps, or at least in $\poly(\norm{Ht}, \log N)$ steps. Indeed, if such simulations existed they could be applied to give new quantum algorithms for various problems~\cite{Chi08}.  Unfortunately, in this section, we show that such simulations are not possible in general.

The currently known dense Hamiltonian simulation algorithms rely on certain properties of the Hamiltonian that we call its \emph{structural properties}. By this we mean the location of nonzero entries in $H$, which correspond to the location of edges in the graph of the Hamiltonian, and the magnitudes of the edge weights. (The remaining information about the Hamiltonian is the phase of each matrix entry $H_{ij}$.)

Given the structural information, the currently known algorithms can simulate dense Hamiltonians using $\O(\norm{\abs(Ht)})$ or $\O(\norm{Ht}_1)$ calls to an oracle that gives the phase of the matrix entry at $H_{ij}$. We call this the \emph{matrix entry phase oracle}. This oracle provides the value of $H_{ij}/|H_{ij}|$ when queried with the input $(i,j)$.  (The oracle may return any complex number of unit modulus---say, $1$---when $H_{ij}=0$.)

We show that given this matrix entry phase oracle and complete structural information, there exist some Hamiltonians that cannot be simulated with $\poly(\norm{Ht}, \log N)$ queries (although they can be simulated with  $\O(\norm{\abs(Ht)})$ queries~\cite{Chi08}). The following theorem is our main result.

\begin{theorem}\label{thm:densehard}
No quantum algorithm can simulate a general Hamiltonian $H \in \C^{N \times N}$ for time $t$ with $\poly(\norm{Ht},\log N)$ queries to a matrix entry phase oracle, even when given complete structural information about the Hamiltonian. 
\end{theorem}

\begin{proof}
The proof of the theorem is divided into two parts. First we show that there exists a set of Hamiltonians of size $N \times N$ that is hard to simulate on average for a particular time $t$. Specifically, we show that simulating a Hamiltonian selected uniformly at random from this set for a chosen time has average-case query complexity $\Omega(\sqrt{N/\log N})$. Then we show that a Hamiltonian simulation algorithm that makes $\poly(\norm{Ht},\log N)$ queries would violate this lower bound. The lemmas used in the proof are proved in the appendix.

To show the lower bound, we need a black-box problem with an $\Omega(\sqrt{N/\log N})$ average-case lower bound, and a set of Hamiltonians whose simulation would solve this problem. We consider the problem of distinguishing strings $s \in \{-1,+1\}^M$ that have sum $-B$ or $+B$, given a black box for the entries of the string. When queried with an index $i \in \{1,2,\ldots,M\}$, the black box returns the value of $s_i \in \{-1,+1\}$, where $s=s_1 s_2 \ldots s_M$. The following lemma characterizes the query complexity of this problem.

\begin{restatable}{lemma}{advbound}
\label{lem:advbound}
Suppose we are given black-box access to a string $s \in \{-1,+1\}^M$, where $s$ is chosen uniformly at random from the set of strings with $\sum_i s_i \in \{-B,+B\}$.  Then determining $\sum_i s_i$ has average-case quantum query complexity $\Theta(M/B)$.
\end{restatable}

Thus, determining whether the sum is $-\sqrt{M \log M}$ or $+\sqrt{M \log M}$, with the promise that one of these is the case, requires $\Omega(\sqrt{M/\log M})$ quantum queries on average. For each string $s$, we construct a Hamiltonian $H_s$ whose simulation for a particular time allows us to distinguish the two possible cases assuming $s$ satisfies the promise. 

Let $H_s$ be a symmetric circulant matrix of size $N \times N$, where $N=2M+1$ is odd. (A circulant matrix is a matrix in which each row is rotated one element to the right relative to the preceding row.) In general, a circulant matrix is completely specified by its first row. However, since $H_s$ is a symmetric circulant matrix, it is completely specified by the first $M+1$ entries of the first row. Let the first entry of the first row be 0, and the next $M$ entries of the first row be $s_1,s_2,\ldots,s_M$. In other words, the first $M+1$ entries of the first row of $H_s$ are $0$ followed by the string $s$.  Then the remaining entries of the first row are $s_M,s_{M-1},\ldots,s_1$.

Given a black box for the entries of $s$, we can easily construct a black box for the entries of $H_s$. Indeed, one query to $H$ can be simulated with at most one query to the string $s$. Sometimes no query to $s$ is needed, since the diagonal entries of $H_s$ are always 0.

Since $H_s$ is a circulant matrix, it is diagonalized by the discrete Fourier transform. Its eigenvalues $\lambda_{0},\lambda_{1},\ldots,\lambda_{N-1}$ are
\be
\lambda_{k}=2\sum_{j=1}^{M}s_{j}\cos\left(\frac{2\pi jk}{N}\right), \quad  \textrm{and in particular,} \quad 
\lambda_{0}=2\sum_{j=1}^{M}s_{j}.
\ee
Thus the time evolution of $H_s$ can be used to learn whether $\sum_j s_j$ is $-\sqrt{M \log M}$ or $+\sqrt{M \log M}$. Since $\lambda_0=2\sum_j s_j$, the two cases can be distinguished by determining the sign of $\lambda_0$. Note that we know the eigenvector corresponding to $\lambda_0$: it is the first column of the discrete Fourier transform matrix, i.e., the uniform superposition over all computational basis states. 

Consider the eigenvalues and eigenvectors of the unitary matrix $e^{-\ii H\tau}$ that corresponds to evolving $H$ for time $\tau={\pi}/{4\sqrt{M\log M}}$. The eigenvectors of this matrix are the same as those of $H$, and each eigenvalue $\lambda_k$ of $H$ corresponds to the eigenvalue $\exp\(-\ii\lambda_k\tau\)$ of $e^{-\ii H\tau}$. Thus the uniform superposition is an eigenvector of $e^{-\ii H\tau}$ with eigenvalue $\exp\(-\ii\lambda_0\tau\)=\exp\(-\ii\pi \sum_i s_i/2\sqrt{M\log M}\)$. Since $\sum_i s_i/\sqrt{M\log M}$ is either $\pm 1$, the two possible eigenvalues are $\pm\ii$. Because the eigenvector is known, the two possibilities can be easily distinguished by phase estimation on the unitary $e^{-\ii H\tau}$. Since the problem of distinguishing these two cases has an $\Omega(\sqrt{M/\log M})$ average-case lower bound by \lem{advbound}, we get an $\Omega(\sqrt{M/\log M})$ average-case lower bound for simulating such Hamiltonians for time $\tau={\pi}/{4\sqrt{M\log M}}$.

Now we want to show that a $\poly(\norm{Ht},\log N)$ Hamiltonian simulation algorithm violates this lower bound. To do this, we need to know the typical behavior of $\norm{H_s}$ when $s$ satisfies the promise. Let $\mathcal{S} \defeq \{-1,+1\}^M$ and let $\mathcal{P}$ be the subset of strings in $\mathcal{S}$ that satisfy the promise. As a first step, let us see the behavior of $\norm{H_s}$ for all strings $s\in \mathcal{S}$, not just those that satisfy the promise. 

\begin{restatable}{lemma}{smallnorm}
\label{lem:smallnorm}
Let $H_s \in \R^{N \times N}$ be a symmetric circulant matrix of size $N=2M+1$ with the first $M+1$ entries of the first row given by $0$ followed by a string $s \in \mathcal{S}$. If $s$ is chosen uniformly at random from $\mathcal{S}$, denoted $s\inuar \mathcal{S}$, then for any $d>0$, 
\be
\Pr_{s \inuar \mathcal{S}}\left(\norm{H_s}\geq 4d\sqrt{M\log M}\right)\leq\frac{4+o(1)}{M^{2d^2-1}}.
\ee
\end{restatable}

In fact, even stronger results of this kind are known~\cites{SZ54,Hal73}, but the above bound is easy to prove and sufficient for our purposes. Using \lem{smallnorm}, we wish to bound the spectral norm of $H_s$ when $s \inuar \mathcal{P}$. We can do so by first calculating the probability that a randomly selected string satisfies the promise. 

\begin{restatable}{lemma}{promisedset}
\label{lem:promisedset}
If $s \inuar \mathcal{S}$, then the probability that $s$ satisfies the promise is
\be
\Pr_{s \inuar \mathcal{S}}(s \in \mathcal{P})=\Theta(1/M).
\ee
\end{restatable}

Using Lemmas \ref{lem:smallnorm} and \ref{lem:promisedset}, we can upper bound the probability that $\norm{H_s}$ is large when $s$ is chosen uniformly at random from $\mathcal{P}$. If $X$ is the event that $\norm{H_s}\geq 4d\sqrt{M\log M}$ and $Y$ is the event that $s \in \mathcal{P}$, then $\Pr(X)$ is given by \lem{smallnorm} and $\Pr(Y)$ is given by \lem{promisedset}.   In these terms, we can compute an upper bound for $\Pr(X|Y)$ as follows: 
\be
\Pr_{s\inuar \mathcal{P}}\(\norm{H_s}\geq 4d\sqrt{M\log M}\)
=\Pr(X|Y)
=\frac{\Pr(X \cap Y)}{\Pr(Y)}
\leq \frac{\Pr(X)}{\Pr(Y)}
=O(M^{2-2d^2}).
\label{eq:largenormprob}
\ee

To achieve a contradiction, assume that for some constant $c>0$, there exists a Hamiltonian simulation algorithm using $O\((\norm{Ht}\log M)^c\)$ queries to simulate $H$ for time $t$. By \eq{largenormprob}, we know that $H_s$ almost always has spectral norm smaller than $4d\sqrt{M\log M}$ when $s \inuar \mathcal{P}$. Since $\norm{H_s} \le \norm{H_s}_1 = 2M$, we can compute the average-case query complexity of the claimed algorithm as follows:
\ba
\mathop {\mathbb E}_{s \inuar \mathcal{P}}\left(\norm{Ht}^c\right) 
&\leq \Pr_{s\inuar \mathcal{P}}\(\norm{H_s} < 4d\sqrt{M\log M}\) O\(\left( 4d\sqrt{M\log M} \frac{\pi}{4\sqrt{M\log M}}\log M\right)^c\) \nonumber\\
&\quad+  \Pr_{s\inuar \mathcal{P}}\(\norm{H_s}\geq 4d\sqrt{M\log M}\) O\(\left((2M)\frac{\pi}{4\sqrt{M\log M}}\log M\right)^c\)\\
&\leq O((d\log M)^c) + O\({M^{c/2-2d^2+2}}{(\log M)^{c/2}}\),
\ea
and by choosing $2d^2>c/2+2$, we have
\ba
\mathop {\mathbb E}_{s \inuar \mathcal{P}}\left(\norm{Ht}^c\right) 
&= O\((\log M)^c\).
\ea
Thus the average-case query complexity of the claimed algorithm is $O\((\log M)^c\)$, which violates the lower bound of  $\Omega(\sqrt{M/\log M})$. 
\end{proof}

The proof technique above can be extended to rule out algorithms with query complexity sub-exponential in $(\norm{Ht},\log N)$ as well, by changing the promised set (i.e., the value of $B$ used in \lem{advbound}) and choosing a larger value of $d$ in \lem{smallnorm}. Exponential functions of $(\norm{Ht},\log N)$ cannot be ruled out, of course, since any Hamiltonian can be simulated by making $O(N^2)$ queries, which is exponential in $\log N$. On the other hand, if we insist that the query complexity of an algorithm depends only on $\norm{Ht}$ (and not $\log N$), then the proof above can be modified to rule out algorithms whose time complexity is an arbitrary function of $\norm{Ht}$. For example, there exists no Hamiltonian simulation algorithm that makes $\exp(\exp(\norm{Ht}))$ queries. 

Finally, we emphasize that even though the above proof involves average-case complexity and distributions over inputs, \thm{densehard} is a statement about the worst-case complexity of simulating Hamiltonians.

\section{Simulation complexity for structured Hamiltonians}
\label{sec:structured}

As the previous section shows, we cannot hope for general Hamiltonian simulation algorithms that scale polynomially in the spectral norm of the Hamiltonian. Although we do know algorithms that scale like $\O(\norm{\abs(H)t})$, \lem{norms} tells us that $\norm{\abs(H)}$ could be exponentially larger than $\norm{H}$. However, we can achieve better scaling for special classes of Hamiltonians.  For example, we saw in \sec{measures} that much stronger bounds hold for sparse Hamiltonians.

We can also improve the inequalities of \lem{norms} for certain classes of non-sparse Hamiltonians.  For example, consider the class of Hamiltonians whose graphs are trees (where the graph of a matrix refers to the graph of its nonzero entries).  Such Hamiltonians can be efficiently simulated even when they are not sparse: in this case, \thm{dense} gives a simulation using $\O(\norm{Ht})$ steps of a discrete-time quantum walk, because when the graph of $H$ is a tree, $\norm{\abs(H)}=\norm{H}$.

\begin{proposition} \label{prop:tree}
If the graph of a Hermitian matrix $H$ is a tree, then there exists a unitary matrix $U$ such that $UHU^{\dag}=\abs(H)$.  In particular, $\norm{\abs(H)} = \norm{H}$. 
\end{proposition}

\begin{proof}
The matrix $U$ is diagonal.  To define $U_{ii}$, we arbitrarily fix some vertex as the root and consider the unique path from the root to vertex $i$.  Let the path contain the vertices $i_0, i_1, \ldots , i_{p-1}, i_p, i$, where $i_0$ is the root and $i_p$ is the parent of $i$.  For each nonzero entry of $H$, define $\alpha_{ij}\defeq H_{ij}/|H_{ij}|$.  Then let $U_{ii}\defeq 1$ if $i$ is the root and $U_{ii} \defeq \alpha_{i_0 i_1}\alpha_{i_1 i_2} \cdots  \alpha_{i_{p-1} i_p} \alpha_{i_p i}$ otherwise.

Since $U$ is diagonal, $(UHU^\dag)_{ij} = U_{ii} H_{ij} U_{jj}^*$.  If $i$ and $j$ are not adjacent in the tree, then $(UHU^\dag)_{ij} = H_{ij} = 0$ as required.  Otherwise, suppose without loss of generality that $j$ is the parent of $i$.  Then $U_{ii}U_{jj}^* = |\alpha_{i_0 i_1}|^2 |\alpha_{i_1 i_2}|^2 \cdots  |\alpha_{i_{p-1} i_p}|^2 \alpha_{j i} = H_{ji}/|H_{ij}|$, so $(UHU^\dag)_{ij} = |H_{ij}|$ as claimed.
\end{proof}

Thus \thm{dense} gives a simulation using $\O(\norm{Ht})$ steps of a discrete-time quantum walk. However, there is another simulation method for such Hamiltonians that uses only $\mcn(Ht)^{1+\o(1)}$ steps~\cite[Theorem~4]{Chi08}. It seems from \lem{norms} that the former method might be inferior, but due to \prop{arb} below it is, in fact, superior to the $\mcn(Ht)^{1+\o(1)}$ simulation (except with respect to error scaling), since $\norm{Ht} \leq 2 \mcn(Ht)$ when the graph of the Hamiltonian $H$ is a tree. 

If $H$ can be expressed as the sum of a small number of Hamiltonians, each of whose graph is a forest, then $H$ can be efficiently simulated when $\norm{H}$ is small.   Recall that a graph is said to have arboricity $k$ if its adjacency matrix can be written as the sum of the adjacency matrices of $k$ forests, but not $k-1$ forests.

\begin{proposition}\label{prop:arb}
If the graph of a Hamiltonian $H$ has arboricity $k$, then $\norm{\abs(H)} \le 2k \mcn(H)$. Moreover, $\norm{\abs(H)} \le 2k \norm{H}$ and $\norm{H} \le 2k \mcn(H)$. 
\end{proposition}

\begin{proof}
We begin by considering the case of a star graph. A star graph is a tree on $n$ vertices with one vertex having degree $n-1$ and the others having degree $1$ (i.e., the complete bipartite graph $K_{1,n-1}$). We show that if $S$ is a Hamiltonian whose graph is a star, 
\be
\mcn(S) = \norm{S} = \norm{\abs(S)}.
\ee

By permuting the vertices, the first vertex can be chosen to be the one with maximum degree. Now the first column of the matrix $S$ completely determines the Hamiltonian. Let the first column be $w$. The matrix $S$ has first column $w$ and first row $w^{\dag}$. It is easy to see that $\mcn(S)=\norm{w}$. $S$ has exactly two nonzero eigenvalues, $\pm\norm{w}$, corresponding to the eigenvectors $\norm{w}e_1 \pm w$, where $e_1$ is the first column of the identity matrix. Since $\norm{S}$ is the maximum eigenvalue, $\norm{S}=\mcn(S)$. 

Moreover, since $\abs(S)$ is a Hamiltonian whose graph is a star, we have  $\norm{\abs(S)}=\mcn(\abs(S))$. For any matrix $H$, $\mcn(\abs(H))=\mcn(H)$, since the norms of the columns depend only on the magnitude of each entry. This proves the desired result, $\mcn(S) = \norm{S} = \norm{\abs(S)}$. These results also hold for a disjoint union of star graphs (a forest of stars), since the above norms all have the property that $\nu(A_1\oplus \ldots \oplus A_n)=\max(\nu(A_1), \ldots , \nu(A_n))$.

To show the result for graphs of arboricity $k$, we begin by showing how a rooted tree can be decomposed into the sum of two forests of stars. The first forest contains all the edges in which the parent vertex is at a even distance from the root. The second forest contains the rest of the edges. This decomposes a rooted tree into two forests of stars, and similarly decomposes a forest into two forests of stars. Since the Hamiltonian has arboricity $k$, it can be decomposed into $k$ forests, which can be decomposed into $2k$ forests of stars. 

Thus $H=\sum_{l=0}^{2k} S_l$, where each of the $S_l$ is a Hermitian matrix whose graph is a forest of stars. Moreover, the matrices $S_l$ have no overlapping edges, i.e., if $(S_l)_{ij} \neq 0$ for some $l$, then $(S_l)_{ij}=0$ for all other $l$. Therefore, for all $i,j,l$, $H_{ij} \geq (S_l)_{ij}$, which implies $\mcn(H)\geq \mcn(S_l)$ for all $l$. This gives
\be
\mcn(H) \geq \frac{1}{2k} \sum_l \mcn(S_l) = \frac{1}{2k} \sum_l \norm{S_l}.
\label{eq:normlower}
\ee
Using the triangle inequality, we find
\be
\norm{\abs(H)} \leq \sum_l \norm{\abs(S_l)} = \sum_l \norm{S_l}.
\label{eq:absnormupper}
\ee

Combining \eq{normlower} and \eq{absnormupper} gives the main result. Using \lem{norms} and the main result gives $\norm{\abs(H)} \le 2k \norm{H}$ and $\norm{H} \le 2k \mcn(H)$. 
\end{proof}

\section{Open questions}
\label{sec:open}

Although we have ruled out the possibility of a generic Hamiltonian simulation algorithm using only $\poly(\norm{Ht}, \log N)$ operations, we can nevertheless hope that some nontrivial classes of Hamiltonians can be simulated in $\poly(\norm{Ht}, \log N)$ steps even though $\norm{H} \ll \norm{\abs(H)}$.

One approach is to consider changing the basis in which the Hamiltonian is simulated.  Clearly, if unitary transformations $U$ and $U^{\dag}$ can be performed efficiently, then $H$ can be simulated efficiently if and only if $UHU^{\dag}$ can.  There must exist bases in which $UHU^{\dag}$ is sparse (such as the basis in which it is diagonal), which may lead to efficient simulations of $H$.  Some trivial classes of Hamiltonians can be simulated in this way, such as Hamiltonians that are tensor products of small factors.  For example, the Hamiltonian $R^{\otimes n}$, where $R \defeq \left(\begin{smallmatrix}1 & 1 \\ 1 & -1\end{smallmatrix}\right)/\sqrt2$ is the Hadamard matrix, has $\norm{R^{\otimes n}}=1$ and $\norm{R^{\otimes n}}_1=\norm{\abs(R^{\otimes n})}=2^{n/2}$, yet the evolution according to $R^{\otimes n}$ is easy to simulate.  A similar simulation for a case where the Hamiltonian is not a tensor product was used in Ref.~\cite{CSV07}.

An alternative method is to investigate ways of decomposing a Hamiltonian as a sum of Hamiltonians that can be efficiently simulated.  For example, we can simulate Hamiltonians whose graphs have polynomial arboricity by decomposing them into stars (although we saw in \prop{arb} that such Hamiltonians can already by simulated efficiently by the method of Ref.~\cite{Chi08} since they satisfy $\norm{\abs(H)}=\poly(\norm{H})$).  More generally, other graph decompositions could give rise to new efficient simulations.

Another interesting problem is to find classes of Hamiltonians that can be simulated in sublinear time. These correspond to quantum systems whose time evolution can be fast-forwarded. Some Hamiltonians may even be simulated in constant time, if $e^{-\ii H\tau}=I$ after some constant time $\tau$ (for example,  the case $R^{\otimes n}$ mentioned above has $\tau=2\pi$).

\section*{Acknowledgments}

We thank Aram Harrow for suggesting the use of Hoeffding's inequality to simply the proof of \lem{smallnorm}.
This work was supported by MITACS, NSERC, QuantumWorks, and the US ARO/DTO.

\appendix
\section*{Appendix: Proofs of lemmas}

In this appendix, we prove Lemmas \ref{lem:advbound}, \ref{lem:smallnorm}, and \ref{lem:promisedset}.

\advbound*

\begin{proof}
We show the lower bound by first showing the same lower bound for the worst-case problem using the quantum adversary method~\cite{Amb02} and then reducing the worst-case problem to the average-case problem. 

For the worst-case lower bound, we use the notation of Theorem 2 of Ref.~\cite{Amb02}. We require two sets of inputs $X$ and $Y$ that have different outputs. Let $X$ be the set of all strings for which $\sum_i s_i = -B$, and $Y$ be the set for which $\sum_i s_i = +B$. We define a relation between the sets as follows. Let an element $x \in X$ be related to an element of $y \in Y$ if and only if $y$ can be reached from $x$ by changing exactly $B/2$ $-1$s to $+1$s in the string $x$. Note that a string in $X$ has exactly $(M/2 + B/2)$ $-1$s and $(M/2 - B/2)$ $+1$s.

Using these sets $X$ and $Y$, and the relation defined above, it is easy to see that
\be 
m=m'=\binom{M/2 + B/2}{B} \quad \textrm{and} \quad l=l'=\binom{M/2 + B/2 -1}{B-1},
\ee
so
\be
  \frac{m}{l} = \frac{m'}{l'} = \frac{1}{2}\(\frac{M}{B}+1\).
\ee
The adversary method now provides a lower bound of $\Omega\left(\sqrt{{mm'}/{ll'}}\right) = \Omega(M/B)$ for the worst-case query complexity of this problem.

The worst-case query complexity can now be reduced to the average-case query complexity under the uniform distribution over all input strings satisfying the promise. To do this, we first apply a uniformly random permutation to the input string, and then with probability $\frac{1}{2}$ multiply all the entries by $-1$ (and leave them unchanged with probability $\frac{1}{2}$). The resulting distribution is now uniform over all input strings satisfying the promise. If the string is not multiplied by $-1$, then the output of the permuted string is the same as the input string. If all the entries are multiplied by $-1$, then the output of the modified string is the opposite of that of the original input.

The lower bound is tight due to a matching upper bound provided by the algorithm for approximate quantum counting~\cite{BHMT02}. To distinguish the two types of inputs, we can approximately count the number of $+1$s to accuracy $\epsilon = {B}/{2M}$, which requires $O(M/B)$ queries.
\end{proof}

\smallnorm*

\begin{proof} 
The eigenvalues of $H_s$ are $\lambda_{r}=2\sum_{j=1}^{M} s_j \cos\frac{2\pi jr}{N}$, where $r \in \{0,1,\ldots,N-1\}$.  We wish to bound the probability that $\lambda_r$ is large, so as to bound the probability of $\norm{H_s}=\max_r |\lambda_r|$ being large. This is achieved by applying Hoeffding's inequality~\cite[Theorem~2]{Hoe63}.

\begin{theorem}[Hoeffding's inequality]
If $X_1, X_2, \ldots, X_M$ are independent and $a_j \leq X_j \leq b_j$ for all $1\leq j\leq M$, then for any $t>0$, we have
\be
\Pr\left(X-\mathbb{E}\left(X\right) \geq Mt \right)\leq \exp\left(\frac{-2M^2t^2}{\sum_{j=1}^{M}(b_j-a_j)^2}\right)
\ee
where $X = \sum_{j=1}^{M} X_j$.
\end{theorem}

If we take $X_j = 2 s_j\cos\frac{2\pi jr}{N}$, then $X = \lambda_r = 2\sum_{j=1}^{M} s_j\cos\frac{2\pi jr}{N}$, each $X_j$ is between $-2$ and $+2$, and $\mathbb{E}\left(X\right) = \sum_{j=1}^M \mathbb{E}\left(X_j\right) = 0$. By choosing $t = 4d\sqrt{{\log M}/M}$, we get
\be
\Pr\left(\lambda_r \geq 4d\sqrt{M \log M} \right) \leq \exp\left(\frac{-2M^2(16d^2 \log M/M)}{\sum_{j=1}^M 4^2}\right) = \frac{1}{M^{2d^2}}.
\ee
Since a similar inequality holds when $X_j$ is replaced by $-X_j$, we get
\be
\Pr\left({\left|\lambda_{r}\right|}\geq 4d\sqrt{M\log M}\right)\leq\frac{2}{M^{2d^2}}.
\ee
Finally, since $\norm{H_s}=\max_r |\lambda_r|$, a union bound gives
\be
\Pr\left(\norm{H_s}\geq 4d\sqrt{M\log M}\right)\leq\frac{2N}{M^{2d^2}},
\ee
which implies the desired result.
\end{proof}

\promisedset*

\begin{proof}
Of the $2^M$ strings of length $M$, those with sum $-\sqrt{M \log M}$ or $+\sqrt{M \log M}$ have either $\frac{1}{2}(M+\sqrt{M\log M})$ $+1$s or $\frac{1}{2}(M+\sqrt{M\log M})$ $-1$s. Thus the total number of such strings is
\be \label{eq:prom}
2\binom{M}{\frac{M+\sqrt{M\log M}}{2}}.
\ee
We can asymptotically approximate this expression using a well-known approximation for the binomial coefficients (see for example equations 4.5 and 4.10 of Ref.~\cite{Odl95}),
which states that
\be
\binom{n}{k} \sim \frac{2^n \exp\(-2(k-n/2)^2/n\)}{\sqrt{\pi n/2}}
\ee
provided $|k-n/2| = o(n^{2/3})$.
Applying this to \eq{prom}, we get
\be
2\binom{M}{\frac{M+\sqrt{M\log M}}{2}}=\Theta\(2^M\exp(-\log M/2)/\sqrt{M}\) = \Theta(2^M/M),
\ee
which proves the claim.
\end{proof}


\end{document}